\documentclass[a4paper,UKenglish,cleveref, autoref, thm-restate]{lipics-v2021}

\newtheorem{rrule}{Reduction-Rule} 
\usepackage{tcolorbox}

\title{Solving Co-Path/Cycle Packing and Co-Path Packing Faster Than $3^k$} 

\titlerunning{Solving Co-Path/Cycle Packing and Co-Path Packing Faster than $3^k$} 

\author{Yuxi Liu}{University of Electronic Science and Technology of China, Chengdu, China}{202211081321@std.uestc.edu.cn}{}{}

\author{Mingyu Xiao}{University of Electronic Science and Technology of China, Chengdu, China}{myxiao@uestc.edu.cn}{https://orcid.org/0000-0002-1012-2373}{}

\authorrunning{J. Open Access and J.\,R. Public}

\Copyright{J. Open Access and J.\,R. Public} 



\category{} 

\relatedversion{} 

\acknowledgements{I want to thank \dots}

\nolinenumbers 

\EventEditors{John Q. Open and Joan R. Access}
\EventNoEds{2}
\EventLongTitle{42nd Conference on Very Important Topics (CVIT 2016)}
\EventShortTitle{CVIT 2016}
\EventAcronym{CVIT}
\EventYear{2016}
\EventDate{December 24--27, 2016}
\EventLocation{Little Whinging, United Kingdom}
\EventLogo{}
\SeriesVolume{42}
\ArticleNo{23}

\begin{document}

\maketitle

\begin{abstract}
    The \textsc{Co-Path/Cycle Packing} problem (resp. The \textsc{Co-Path Packing} problem) asks whether we can delete at most $k$ vertices from the input graph such that the remaining graph is a collection of induced paths and cycles (resp. induced paths).
    These two problems are fundamental graph problems that have important applications in bioinformatics.
    Although these two problems have been extensively studied in parameterized algorithms, it seems hard to break the running time bound $3^k$.
    In 2015, Feng et al. provided an $O^*(3^k)$-time randomized algorithms for both of them. Recently, Tsur showed that they can be solved in $O^*(3^k)$ time deterministically.
    In this paper, by combining several techniques such as path decomposition, dynamic programming, cut \& count, and branch-and-search methods, we show that \textsc{Co-Path/Cycle Packing} can be solved in $O^*(2.8192^k)$ time deterministically and \textsc{Co-Path Packing} can be solved in $O^*(2.9241^{k})$ time with failure probability $\leq 1/3$.
    As a by-product, we also show that the \textsc{Co-Path Packing} problem can be solved in $O^*(5^p)$ time with probability at least 2/3 if a path decomposition of width $p$ is given. 
\end{abstract}

\section{Introduction}
In the classic \textsc{Vertex Cover} problem, the input is a graph $G$ and an integer $k$, and the problem asks whether it is possible to delete at most $k$ vertices such that the maximum degree of the remaining graph is at most 0.
A natural generalization of \textsc{Vertex Cover} is that: can we delete at most $k$ vertices such that the maximum degree of the remaining graph is at most $d$?
Formally, for every integer $d \geq 0$, we consider the following \textsc{$d$-Bounded-Degree Vertex Deletion} problem.

\noindent\rule{\linewidth}{0.2mm}
\textsc{$d$-Bounded-Degree Vertex Deletion}\\
\textbf{Instance:} A graph $G=(V, E)$ and two integers $d$ and $k$.\\
\textbf{Question:} Is there a set of at most $k$ vertices whose removal from $G$ results in a graph with maximum degree at most $d$?\\
\rule{\linewidth}{0.2mm}

The \textsc{$d$-Bounded-Degree Vertex Deletion} problem finds applications in computational biology~\cite{fellows2011generalization} and social network analysis~\cite{seidman1978graph}.
In this paper, we focus on the case that $d = 2$, which is referred to as the \textsc{Co-Path/Cycle Packing} problem.
The \textsc{Co-Path/Cycle Packing} problem also has many applications in computational biology~\cite{chen2010linear}.
Formally, the \textsc{Co-Path/Cycle Packing} problem is defined as follows.

\noindent\rule{\linewidth}{0.2mm}
\textsc{Co-Path/Cycle Packing}\\
\textbf{Instance:} A graph $G=(V, E)$ and an integer $k$.\\
\textbf{Question:} Is there a vertex subset $S\subseteq V$ of size at most $k$ whose deletion makes the graph a collection of induced paths and cycles?\\
\rule{\linewidth}{0.2mm}
  
We also focus on a similar problem called \textsc{Co-Path Packing}, which allows only paths in the remaining graph, defined as follows.

\noindent\rule{\linewidth}{0.2mm}
\textsc{Co-Path Packing}\\
\textbf{Instance:} A graph $G=(V, E)$ and an integer $k$.\\
\textbf{Question:} Is there a vertex subset $S\subseteq V$ of size at most $k$ whose deletion makes the graph a collection of induced paths?\\
\rule{\linewidth}{0.2mm}


\vspace{2mm}
\noindent\textbf{Related Work.}
In this paper, we mainly consider parameterized algorithms.
When $d$ is an input, the general \textsc{$d$-Bounded-Degree Vertex Deletion} problem is W[2]-hard with the parameter $k$~\cite{fellows2011generalization}.
Xiao~\cite{xiao2016parameterized} gave a deterministic algorithm that solves \textsc{$d$-Bounded-Degree Vertex Deletion} in $O^*((d+1)^k)$ time for every $d\geq 3$, which implies that the problem is FPT with parameter $k+d$.
In term of treewidth (tw), van Rooij~\cite{van2021generic} gave an $O^*((d + 2)^{tw})$-time algorithm to solve \textsc{$d$-Bounded-Degree Vertex Deletion} for every $d\geq 1$.
Lampis and Vasilakis~\cite{lampis2023structural} showed that no algorithm can solve \textsc{$d$-Bounded-Degree Vertex Deletion} in time $(d + 2 - \epsilon)^{tw}n^{O(1)}$, for any $\epsilon > 0$ and for any fixed $d\geq 1$ unless the SETH is false.
The upper and lower bounds have matched.
This problem has also been extensively studied in kernelization. Fellows et al.~\cite{fellows2011generalization} and Xiao~\cite{xiao2017generalization} gave a generated form of the NT-theorem that then provided polynomial kernels for the problem with each fixed $d$.

For each fixed small $d$, \textsc{$d$-Bounded-Degree Vertex Deletion} has also been paid certain attention.
The \textsc{$0$-Bounded-Degree Vertex Deletion} problem is referred to as \textsc{Vertex Cover}, which is one of the most fundamental problems in parameterized algorithms.
For a long period of time, the algorithm of Chen et al.~\cite{chen2010improved} held the best-known running time of $O^*(1.2738^k)$,
and recently this result was improved by Harris and Narayanaswamy~\cite{harris2022faster} to $O^*(1.25284^k)$. 
The \textsc{$1$-Bounded-Degree Vertex Deletion} problem is referred to as \textsc{$P_3$ Vertex Cover}, where Tu~\cite{tu2015fixed} achieved a running time of $O^*(2^k)$ by using iterative compression.
This result was later improved by Katreni{\v{c}}~\cite{katrenivc2016faster} to $O^*(1.8127^k)$.
Then, Chang et. al.~\cite{chang2016fixed} gave an $O^*(1.7964^k)$-time polynomial-space algorithm and an $O^*(1.7485^k)$-time exponential-space algorithm.
Xiao and Kou~\cite{xiao2017kernelization} gave an $O^*(1.7485^k)$-time polynomial-space algorithm.
This result was improved by Tsur~\cite{tsur2019parameterized} to $O^*(1.713^k)$ through a branch-and-search approach and finally by {\v{C}}erven{\`y} and Such{\`y}~\cite{vcerveny2023generating} to $O^*(1.708^k)$ through using an automated framework for generating parameterized branching algorithms.

\textsc{Co-Path/Cycle Packing} is the special case of \textsc{$d$-Bounded-Degree Vertex Deletion} with $d = 2$.
A closely related problem is \textsc{Co-Path Packing}, where even cycles are not allowed in the remaining graph.
Chen et al.~\cite{chen2010linear} initially showed that \textsc{Co-Path/Cycle Packing} and \textsc{Co-Path Packing} can be solved in $O^*(3.24^k)$ time, and a finding subsequently refined to $O^*(3.07^k)$ for \textsc{Co-Path/Cycle Packing} by Xiao~\cite{xiao2016parameterized}.
Feng et al.~\cite{feng2015randomized} introduced a randomized $O^*(3^k)$-time algorithm for the \textsc{Co-Path Packing}, which also works for \textsc{Co-Path/Cycle Packing}. However, we do not know how to derandomize this algorithm.
Recently, Tsur~\cite{tsur2022faster} provided $O^*(3^k)$-time algorithms solving \textsc{Co-Path/Cycle Packing} and \textsc{Co-Path Packing} deterministically.
It seems that the bound $3^k$ is hard to break for the two problems.
As shown in Tsur's algorithms~\cite{tsur2022faster}, many cases, including the case of handling all degree-4 vertices, lead to the same bottleneck.
One of the main targets in this paper is to break those bottlenecks. 
Previous results and our results for
\textsc{Co-Path/Cycle Packing} and \textsc{Co-Path Packing} are summarized in Table~\ref*{related-work}.

\begin{table}[!t]
    \begin{center}
    \caption{Algorithms for \textsc{Co-Path/Cycle Packing} and \textsc{Co-Path Packing}}
    \begin{tabular}{c|c|c|c|c|c}
        \hline
        Years & References & \textsc{Co-Path/Cycle} & Deterministic & \textsc{Co-Path} & Deterministic \\
        \hline
        2010 & Chen et al.~\cite{chen2010linear}  & $O^*(3.24^k)$ & Yes & $O^*(3.24^k)$ & Yes\\
        2015 & Feng et al.~\cite{feng2015randomized} & $O^*(3^k)$ & No & $O^*(3^k)$ & No\\
        2016 & Xiao~\cite{xiao2016parameterized} & $O^*(3.07^k)$ & Yes & - & -\\
        2022 & Tsur~\cite{tsur2022faster} & $O^*(3^k)$ & Yes & $O^*(3^k)$ & Yes \\
        2024 & This paper & $O^*(2.8192^k)$ & Yes & $O^*(2.9241^{k})$ & No \\
        \hline
    \end{tabular}\label{related-work}
    \end{center}
\end{table}

\noindent\textbf{Our Contributions.}
The main contributions of this paper are a deterministic algorithm for \textsc{Co-Path/Cycle Packing} running in $O^*(2.8192^k)$ time and $O^*(2.5199^{k})$ space and a randomized algorithm for \textsc{Co-Path Packing} running in $O^*(2.9241^{k})$ time and space with failure probability $\leq 1/3$.
To obtain this result, we need to combine path decomposition, dynamic programming, branch-and-search and some other techniques.
In the previous $O^*(3^k)$-time algorithms for both \textsc{Co-Path/Cycle Packing} and \textsc{Co-Path Packing}, many cases,
including dealing with degree-4 vertices and several different types of degree-3 vertices,
lead to the same bottleneck.
It seems very hard to avoid all the bottleneck cases by simply modifying the previous algorithms.
The main idea of the algorithm in this paper is as follows. 
We first design some new reduction and branching rules to handle some good structures of the graph.
After this, we prove that the remaining graph has a small pathwidth and then design an efficient dynamic programming algorithm based on a path decomposition.
Specifically, our algorithm firstly runs the branch-and-search algorithm to handle the degree-$\geq 5$ vertices and the degree-4 vertices adjacent to at least one degree-$\geq 3$ vertex.
In the branch-and-search phase, our algorithm runs in $O^*(2.8192^k)$ time for both \textsc{Co-Path/Cycle Packing} and \textsc{Co-Path Packing}.
When branching steps cannot be applied, we can construct a path decomposition of width at most $2k/3 + \epsilon k$ for any $\epsilon > 0$ and call our dynamic programming algorithm.
The running time and space of the dynamic programming algorithm are bounded by $O^*(2.5199^{k})$ for \textsc{Co-Path/Cycle Packing} and bounded by $O^*(2.9241^{k})$ with failure probability $\leq 1/3$ for \textsc{Co-Path Packing}.
Therefore, the whole algorithm runs in $O^*(2.8192^k)$ time and $O^*(2.5199^{k})$ space for \textsc{Co-Path/Cycle Packing} and runs in $O^*(2.9241^{k})$ time and space with failure probability $\leq 1/3$ for \textsc{Co-Path Packing}.


For the dynamic programming algorithms based on a path decomposition, 
the first algorithm is an $O^*((d + 2)^p)$-time algorithm to solve \textsc{$d$-Bounded-Degree Vertex Deletion} for every $d\geq 1$, where $p$ is the width of the given path decomposition.
This was firstly found in~\cite{van2021generic}.
We also present it in our way to make this paper self-contained.
The second algorithm is designed for \textsc{Co-Path Packing}.
In this algorithm, we use an algorithm framework called cut \& count~\cite{cygan2011solving}.
Given a path decomposition of width $p$, we show that \textsc{Co-Path Packing} can be solved in $O^*(5^p)$ time and space with failure probability $\leq 1/3$.

\vspace{2mm}
\noindent\textbf{Reading Guide.}
Section \ref*{preliminary-section} begins with a review of fundamental definitions and the establishment of notation.
In Section \ref*{cpcp-section} we solve \textsc{Co-Path/Cycle Packing} and in section \ref*{cpp-section} we solve \textsc{Co-Path Packing}.
In Section \ref*{proper-section}, we show that a proper graph has a small pathwidth and present a dynamic programming algorithm for \textsc{Co-Path/Cycle Packing}. 
In Section \ref*{cpcp-algorithm-section}, we give the branch-and-search algorithm for \textsc{Co-Path/Cycle Packing}.
Similarly, in Section \ref*{cpp-dp-section}, we present a randomized dynamic programming algorithm for \textsc{Co-Path Packing}. 
In Section \ref*{cpp-algorithm-section}, we give the branch-and-search algorithm for \textsc{Co-Path Packing}. 
Most of the content of this branching algorithm is the same as the branching algorithm for \textsc{Co-Path/Cycle Packing}.
Due to lack of space, Sections 1-3 can be seen as the short version. 
In Section \ref*{conclusion-section}, we give a conclusion.

The proof of theorems marked $\clubsuit$ is placed in the appendix.

\section{Preliminaries}\label{preliminary-section}


In this paper, we only consider simple and undirected graphs.
Let $G=(V, E)$ be a graph with $n=|V|$ vertices and $m=|E|$ edges.
A vertex $v$ is called a \emph{neighbor} of a vertex $u$ if there is an edge $uv \in E$.
Let $N(v)$ denote the set of neighbors of $v$. For a subset of vertices $X$, let $N(X)=\bigcup_{v\in X}N(v)\setminus X$ and $N[X]=N(X)\cup X$.
We use $d(v)=|N(v)|$ to denote the \emph{degree} of a vertex $v$ in $G$.
A vertex of degree $d$ is called a \emph{degree-$d$ vertex}.
For a subset of vertices $X\subseteq V$, the subgraph induced by $X$ is denoted by $G[X]$. The induced subgraph $G[V\setminus X]$ is also written as $G\setminus X$.
A \emph{path} $P$ in $G$ is a sequence of vertices $v_1, v_2,\cdots, v_t$ such that for any $1\leq i < t, \{v_{i}v_{i+1}\} \in E$.
Two vertices $u$ and $v$ are \emph{reachable} to each other if there is a path $v_1, v_2,\cdots, v_t$ such that $v_1 = u$ and $v_t = v$.
A \emph{connected component} of a graph is a maximum subgraph such that any two vertices are reachable to each other.
A vertex subset $S$ is called a \emph{cPCP-set} of graph $G$ if the degree of any vertex in $G\setminus S$ is at most 2.
A vertex subset $S$ is called a \emph{cPP-set} of graph $G$ if $G\setminus S$ is a collection of disjoint paths.
For a graph $G$, we will use $V(G)$ and $E(G)$ to denote the vertex set and edge set of it, respectively.
A complete graph with 3 vertices is called a \emph{triangle}. A singleton $\{v\}$ may be denoted as $v$.

\vspace{2mm}
\noindent\textbf{Path decomposition.} We will use the concepts of path decomposition and nice path decomposition of a graph.

\begin{definition}[\cite{cygan2015parameterized}]
    A \textit{path decomposition} of a graph $G$ is a sequence $P = (X_1, X_2, \cdots, X_r)$ of vertex subsets $X_i \subseteq V(G)$ $(i \in \{1, 2, \cdots, r\})$ such that:\\
    \emph{(P1)} $\bigcup_{i = 1}^r X_i = V(G)$.\\
    \emph{(P2)} For every $uv \in E(G)$, there exists $l \in \{1, 2, \cdots, r\}$ such that $X_l$ contains both $u$ and $v$.\\
    \emph{(P3)} For every $u \in V(G)$, if $u \in X_i \cap X_k$ for some $i \leq k$, then $u \in X_j$  for all $i \leq j \leq k$. 
\end{definition}
For a path decomposition $(X_1, X_2, \cdots, X_r)$ of a graph, each vertex subset $X_i$ in it is called a \emph{bag}.
The \emph{width} of the path decomposition is $\max_i \{|X_i|\} - 1$.
The \emph{pathwidth} of a graph $G$, denoted by pw($G$), is the minimum possible width of a path decomposition of $G$.
A path decomposition $(X_1, X_2, \cdots, X_r)$ is \textit{nice} if $X_1 = X_r = \emptyset$ and for every $i\in \{1, 2, \cdots, r - 1\}$ there is either a vertex $v\notin X_i$ such that $X_{i + 1} = X_i\cup \{v\}$, or there is a vertex $w\in X_i$ such that $X_{i + 1} = X_i\setminus \{w\}$.
The following lemma shows that any path decomposition can be turned into a nice path decomposition without increasing the width.

\begin{lemma}[\cite{cygan2015parameterized}]\label{nice}
    If a graph $G$ admits a path decomposition of width at most $p$, then it also admits a nice path decomposition of width at most $p$.
    Moreover, given a path decomposition $P = (X_1, X_2, \cdots, X_r)$ of $G$ of width at most $p$, one can in time $O(p^2\cdot \max(r, |V(G)|))$ compute a nice path decomposition of $G$ of width at most $p$.
\end{lemma}

There are also easy ways to reduce the length $r$ of a path decomposition to a polynomial of the graph size. Next, we will also assume that $r$ is bounded by a polynomial of the number of vertices.
In terms of the pathwidth, there is a known bound.

\begin{theorem}[\cite{fomin2009two}]\label{pathwidth}
    For any $\epsilon > 0$, there exists an integer $n_\epsilon$ such that for every graph $G$ with $n > n_\epsilon$ vertices,
    \[
        \text{pw}(G) \leq \frac{1}{6}n_3+\frac{1}{3}n_4 + n_{\geq 5} + \epsilon n,
    \]
    where $n_i$ $i\in \{3, 4\}$ is the number of vertices of degree $i$ in $G$ and $n_{\geq 5}$ is the number of vertices of degree at least 5. Moreover, a path decomposition of the corresponding width can be constructed in polynomial time.
\end{theorem}

\noindent\textbf{Branch-and-Search Algorithm.} For a branch-and-search algorithm,
we use a parameter $k$ of the instance to measure the running time of the algorithm. Let $T(k)$ denote the maximum size of the search tree generated by the algorithm when running on an instance with the parameter no greater than $k$.
Assume that a branching operation generates $l$ branches and the measure $k$ in the $i$-th instance decreases by at least $c_i$. This operation generates a recurrence relation
\[
    T(k)\leq T(k-c_1)+T(k-c_2)+...+T(k-c_l)+1.
\]

The largest root of the function $f(x)=1-\sum_{i=1}^lx^{-c_i}$ is called the \textit{branching factor} of the recurrence.
Let $\gamma$ denote the maximum branching factor among all branching factors in the search tree.
The running time of the algorithm is bounded by $O^*(\gamma^k)$. For more details about analyzing branch-and-search algorithms, please refer to~\cite{kratsch2010exact}.
\section{A Parameterized Algorithm for \textsc{Co-Path/Cycle Packing}}\label{cpcp-section}

In this section, we propose a parameterized algorithm for \textsc{Co-Path/Cycle Packing}.
First, in Section \ref*{proper-section}, we show that \textsc{Co-Path/Cycle Packing} on a special graph class, called \emph{proper graph}, can be quickly solved by using the dynamic programming algorithm based on path decompositions in Theorem \ref*{dp}.
The key point in this section is to bound the pathwidth of proper graphs by $2k/3 + \epsilon k$ for any $\epsilon > 0$.
Second, in Section \ref*{cpcp-algorithm-section}, we give a branch-and-search algorithm that will implement some branching steps on special local graph structures. When there is no good structure to apply our branching rules, we show that the graph must be a proper graph and then the algorithm in Section \ref*{proper-section} can be called directly to solve the problem.

\subsection{Proper Graphs with Small Pathwidth}\label{proper-section}
A graph is called \textit{proper} if it satisfies the following conditions:
\begin{enumerate}
    \item The maximum degree of $G$ is at most 4.
    \item For any degree-4 vertex $v$, all neighbors are of degree at most 2.
    \item For any degree-2 vertex $v$, at least one vertex in $N(v)$ is of degree at least 3. 
    \item Each connected component contains at least 6 vertices.
\end{enumerate}

We are going to solve \textsc{Co-Path/Cycle Packing} on proper graphs first. Next, we try to bound the pathwidth of proper graphs.

\begin{lemma}\label{pw-exist}
    Let $G$ be a proper graph. If $G$ has a $cPCP$-set (resp. $cPP$-set) of size at most $k$, then it holds that
    \begin{equation}|V(G)|\leq 100k    ~~~\mbox{and}~~~        \frac{n_3}{6}+\frac{n_4}{3} \leq \frac{2k}{3},\end{equation}
    where $n_3$ and $n_4$ are the number of degree-3 and degree-4 vertices in $G$, respectively.
\end{lemma}
\begin{proof}

    \begin{figure}[!t]
        \centering
        \includegraphics[scale=0.4]{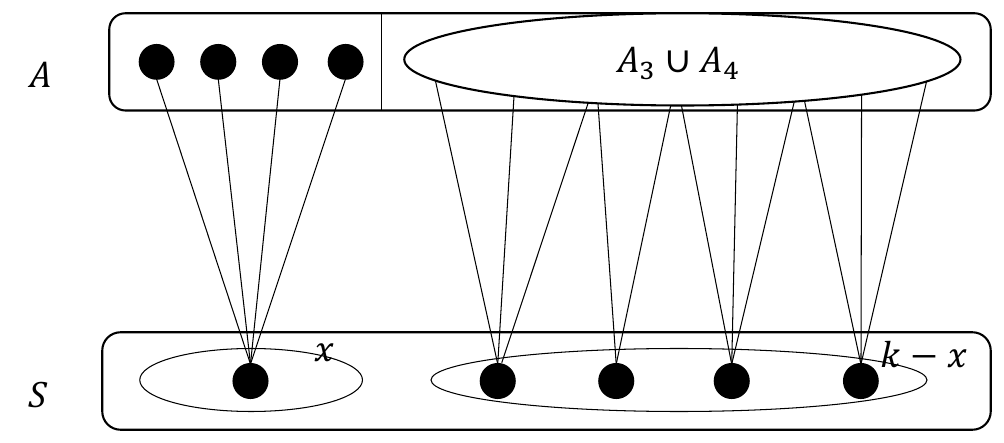}
        \caption{Sets $S$, $A$, $A_3$ and $A_4$ in the proof of Lemma \ref*{pw-exist}.}
        \label{SmallPW}
    \end{figure}

    Assume that there is a cPCP-set (resp. cPP-set) $S$ of size at most $k$. Let $A = V(G) \setminus S$.
    Let $V_3$  and $V_4$ be the set of degree-3 and degree-4 vertices in $G$, respectively. Let $A_3 = A \cap V_3$ and $A_4 = A \cap V_4$.
    Let $x$ be the number of degree-4 vertices in $S$. See Figure \ref*{SmallPW} for an illustration.

    Since the degree of any vertex in $G[A]$ is at most 2, the number of edges between $A_3\cup A_4$ and $S$ is at least $|A_3| + 2|A_4|$.
    Since $G$ is proper, for any degree-4 vertex $v$, there is no degree-$\geq 3$ vertex in $N(v)$.
    So we have that the number of edges between $A_3\cup A_4$ and $S$ is at most $3(k - x)$.
    So we have that
    \begin{equation}
        3(k - x) \geq |A_3| + 2|A_4|. %
    \end{equation}

    Let $n_i$ be the number of degree-$i$ vertices in $G$ for any $i\in \{0, 1, 2, 3, 4\}$.
    By the fourth condition of the proper graph, we have that $n_0 = 0$.
    We also have that $n_3 \leq k - x + |A_3|$, $n_4 = |A_4| + x$.
    Inequality (2) implies that
    \begin{equation}
        n_3 + n_4 \leq k + |A_3| + |A_4| \leq k + |A_3| + 2|A_4| \leq 4k. %
    \end{equation}

    By the third condition of the proper graph, 
    we have that
    \begin{equation}
        n_2 \leq 4(n_3 + n_4) \leq 16k. %
    \end{equation}

    For degree-1 vertices, also by the fourth condition of the proper graph, we have that
    \begin{equation}
        n_1 \leq 4(n_2 + n_3 + n_4) \leq 4(16k + 4k) = 80k. %
    \end{equation}

    Inequalities (3), (4) and (5) together imply that
    \begin{equation}
        |V(G)| = n_0 + n_1 + n_2 + n_3 + n_4 \leq 100k. %
    \end{equation}

    Since $x \geq 0$, inequality (2) implies that
    \[
        \frac{n_3}{6}+\frac{n_4}{3} \leq \frac{k - x + |A_3| }{6}+\frac{|A_4| + x}{3}= \frac{k + x + |A_3| + 2|A_4| }{6}\leq \frac{4k - 2x}{6}\leq \frac{2k}{3}. %
    \]
\end{proof}

\begin{lemma}\label{get-pw}
    Let $G$ be a proper graph. For any $\epsilon > 0$, in polynomial time we can either decide that $G$ has no cPCP-set (resp. cPP-set) of size at most $k$ or compute a path decomposition of width at most $\frac{2k}{3} + \epsilon k$.
\end{lemma}

\begin{proof}

    Our algorithm is defined as follows.
    Let $n_3$ and $n_4$ be the number of degree-3 and degree-4 vertices in $G$, respectively.
    First, we check whether $|V(G)|\leq 100k$ and $\frac{n_3}{6}+\frac{n_4}{3} \leq \frac{2k}{3}$ in polynomial time.
    If $|V(G)| > 100k$ or $\frac{n_3}{6}+\frac{n_4}{3} > \frac{2k}{3}$, by Lemma \ref*{pw-exist}, we can decide that $G$ has no cPCP-set (resp. cPP-set) of size at most $k$.
    Otherwise, let $\epsilon' = \epsilon / 100$. By Theorem \ref*{pathwidth}, we can obtain a path decomposition $P'$ of width at most $p' = \frac{n_3}{6}+\frac{n_4}{3} + \epsilon'|V(G)|\leq \frac{2k}{3} + 100\epsilon' k = \frac{2k}{3} + \epsilon k$.
    This lemma holds.
\end{proof}


The following theorem shows that there exists an algorithm for the general \textsc{$d$-Bounded-Degree Vertex Deletion} problem based on a given path decomposition of the graph.
The running time bound of the algorithm is $O^*((d + 2)^p)$, where $p$ is the width of the given path decomposition.
Previously, an $O^*(3^p)$-time algorithm for \textsc{$1$-Bounded-Degree Vertex Deletion} was known~\cite{chang2016fixed}. 
Recently, this result was extended for any $d\geq1$ by van Rooij~\cite{van2021generic}.
We also present it (in the appendix) in our way to make this paper self-contained. 

\begin{theorem}[$\clubsuit$]\label{dp}
    Given a path decomposition of $G$ with width $p$. For any $d\geq 1$, \textsc{$d$-Bounded-Degree Vertex Deletion} can be solved in $O^*((d+2)^p)$ time and space.
\end{theorem}

However, the algorithm given in Theorem \ref*{dp} cannot be used to solve \textsc{Co-Path Packing} since there is a global connectivity constraint for \textsc{Co-Path Packing}. 
We will discuss it in Section \ref*{cpp-section}.
Based on Theorem \ref*{dp}, we have the following lemma.

\begin{lemma}\label{solve-proper}
    \textsc{Co-Path/Cycle Packing} on proper graphs can be solved in $O^*(2.5199^{k})$ time.
\end{lemma}

\begin{proof}
    We first call the algorithm in Lemma \ref*{get-pw}.
    If the algorithm decides that $G$ has no $cPCP$-set of size at most $k$, we claim that $(G, k)$ is a no-instance.
    Otherwise, we can obtain a nice path decomposition of width at most $\frac{2k}{3} + \epsilon k$.
    Then we call the algorithm in Theorem \ref*{dp}. This algorithm runs in $O^*(4^{2k/3+\epsilon k}) = O^*(2.5199^{k})$, where we choose $\epsilon < 10^{-6}$. 
    This lemma holds.
\end{proof}





\subsection{A Branch-and-Search Algorithm}\label{cpcp-algorithm-section}
In this subsection, we provide a branch-and-search algorithm for \textsc{Co-Path/Cycle Packing}, which is denoted by ${\tt cPCP}(G, k)$.
Our algorithm contains
several reduction and branching steps. After recursively executing these steps, we will get a proper graph and then call the dynamic programming algorithm in Section \ref*{proper-section} to solve it.

\subsubsection{Reduction and Branching Rules}

Firstly we have a reduction rule to reduce small connected components.
\begin{rrule}\label{rrule-1}
    If there is a connected component $C$ of the graph such that $|V(C)|\leq 6$, then run a brute force algorithm to find a minimum cPCP-set $S$ in $C$, delete $C$ and include $S$ in the deletion set.
\end{rrule}

\begin{lemma}\label{reduction-2-lemma}
    Let $u$ and $v$ be two adjacent vertices of degree at most 2 in $G$ and $G'$ be the graph after deleting edge $uv$ from $G$.
    Then $(G, k)$ is a yes-instance if and only if $(G', k)$ is a yes-instance.
\end{lemma}
This lemma holds because adding an edge between two vertices of degree at most 1 back to a graph will not make the two vertices of degree greater than 2.
Based on this lemma, we have the following reduction rule.

\begin{rrule}\label{rrule-2}
    If there are two adjacent vertices $u$ and $v$ of degree at most 2, then
    return ${\tt cPCP}(G' = (V(G), E(G) \setminus \{uv\}), k)$.
\end{rrule}

\begin{lemma}\label{reduction-3-lemma}
    Let $\{u, v, w\}$ be a triangle such that $|N(\{u, v, w\})| = 1$. Let $x$ be the vertex in $N(\{u, v, w\})$.
    There is a minimum cPCP-set containing $x$.
\end{lemma}
This lemma holds because we must delete at least one vertex in $\{u, v, w, x\}$ and delete any vertex in $\{u, v, w\}$ cannot decrease the degree of vertices in $V(G)\setminus \{u, v, w, x\}$.
Based on this lemma, we have the following reduction rule.

\begin{rrule}\label{rrule-3}
    If there is a triangle $\{u, v, w\}$ such that $|N(\{u, v, w\})| = 1$, then
    return ${\tt cPCP}(G\setminus N[\{u, v, w\}], k - 1)$.
\end{rrule}

After applying the three simple reduction rules, we will execute some branching steps.
Although we have several branching steps, most of them are based on the following two branching rules.

For a vertex $v$ of degree at least 3, either it is included in the deletion set or it remains in the graph.
For the latter case, there are at most two vertices in $N(v)$ that can remain in the graph.
So we have the following branching rule.

\vspace{2mm}
\noindent\textbf{Branching-Rule (B1).}
\textit{
    For a vertex $v$ of degree at least 3, branch on it to generate ${|N(v)|\choose 2} + 1$ branches by
    either (i) deleting $v$ from the graph and including it in the deletion set,
    or (ii) for every pair of vertices $u$ and  $w$ in $N(v)$,
    deleting $N(v)\setminus \{u, w\}$ from the graph and including $N(v)\setminus \{u, w\}$ in the deletion set.
}
\vspace{2mm}

A vertex $v$ \textit{dominates} a vertex $u$ if $N[u]\subseteq N[v]$. We have the following property for dominated vertices.

\begin{lemma}\label{domination}
    If a vertex $v$ dominates a vertex $u$, then there is a minimum cPCP-set either containing $v$ or containing none of $v$ and $u$.
\end{lemma}

\begin{proof}
    Let $S$ be a cPCP-set. Assume to the contrary that $S\cap \{v,u\}=\{u\}$.
    Since $v$ dominates $u$, we know that $S' = S\setminus \{u\} \cup \{v\}$ is still a feasible cPCP-set with $|S'|= |S|$.  There is a minimum cPCP-set containing $v$.     This lemma holds.
\end{proof}

Assume that $v$ dominates $u$. By Lemma~\ref{domination}, we know that either $v$ is included in the deletion set or both of $v$ and $u$ remain in the graph.
For the latter case, there are at most two vertices in $N(v)$ that can remain in the graph. Thus, at least $|N(v)| - 2$ vertices in $N(v)\setminus \{u\}$ will be deleted. We get the following branching rule.

\vspace{2mm}
\noindent\textbf{Branching-Rule (B2).}
\textit{
    Assume that a vertex $v$ of degree at least 3 dominates a vertex $u$.
    Branch on $v$ to generate $1+ (|N(v)|-1)=|N(v)|$ branches by
    either (i) deleting $v$ from the graph and including it in the deletion set,
    or (ii) for each vertex $w\in N(v)\setminus \{u\}$,
    deleting $N(v)\setminus \{u, w\}$ from the graph and including $N(v)\setminus \{u, w\}$ in the deletion set.
}


\subsubsection{Steps}

When we execute one step, we assume that all previous steps are not applicable in the current graph anymore.
We will analyze each step after describing it.

\vspace{2mm}
\noindent\textbf{Step 1} (Vertices of degree at least 5).
    If there is a vertex $v$ of $d(v) \geq 5$, then branch on $v$ with Branching-Rule (B1) to generate ${d(v)\choose d(v) - 2} + 1$ branches and return the best of
\[
    \begin{split}
        & {\tt cPCP}(G \setminus \{v\},k - 1) \\
        \mbox{and} \quad & {\tt cPCP}(G \setminus (N(v)\setminus \{u, w\}),k - |N(v)\setminus \{u, w\}|) \mbox{~for each pair $\{u, w\}\subseteq N(v)$}.
    \end{split}
\]

For this step, we get a recurrence
\[
    T (k) \leq T (k - 1) + {d(v)\choose d(v) - 2} \times T (k - (d(v) - 2)) + 1,
\]
where $d(v) \geq 5$.
For the worst case that $d(v) = 5$, the branching factor of it is 2.5445. 
\vspace{2mm}

After Step 1, the graph contains only vertices with degree at most 4.

\vspace{2mm}
\noindent\textbf{Step 2} (Degree-4 vertices dominating some vertex of degree at least 3).
    Assume that there is a degree-4 vertex $v$ that dominates a vertex $u_1$, where $d(u_1)\geq 3$. Without loss of generality, we assume that the other three neighbors of $v$ are $u_2, u_3$ and $u_4$, and $u_1$ is adjacent to $u_2$ and $u_3$ ($u_1$ is further adjacent to $u_4$ if $d(u_1)=4$). See Figure \ref*{Step2} for an illustration. We branch on $v$ with Branching-Rule (B2) and get $|N(v)|$ branches
    $${\tt cPCP}(G \setminus \{v\},k - 1),  {\tt cPCP}(G \setminus \{u_2, u_3\}, k - 2), $$
$${\tt cPCP}(G \setminus \{u_2, u_4\}, k - 2), ~\mbox{and}~ {\tt cPCP}(G \setminus \{u_3, u_4\}, k - 2).$$

    \begin{figure}[!t]
        \centering
        \includegraphics[scale=0.4]{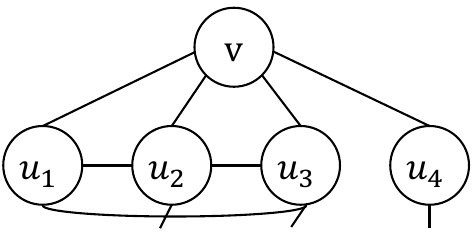}
        \caption{Degree-4 vertex $v$ dominates a degree-3 vertex $u_1$.}
        \label{Step2}
    \end{figure}

The corresponding recurrence is
\[
    T (k) \leq T (k - 1) + 3 \times T (k - 2) + 1,
\]
which has a branching factor of 2.3028.
\vspace{2mm}

A triangle $\{u, v, w\}$ is called \textit{heavy} if it holds that $|N(\{u, v, w\})|\geq 4$.

\vspace{2mm}
\noindent\textbf{Step 3} (Degree-4 vertices in a heavy triangle).
    Assume that a degree-4 vertex $v$ is in a heavy triangle.
    Let $u_1$, $u_2$, $u_3$ and $u_4$ be the four neighbors of $v$, where we assume without loss of generality that $\{v, u_1, u_2\}$ is a heavy triangle. See Figure \ref*{Step3} for an illustration. 
    We branch on $v$ with Branching-Rule (B1).
    In the branch of deleting $\{u_3,u_4\}$, we can simply assume that the three vertices $v, u_1$ and $u_2$ are not deleted, otherwise this branch can be covered by another branch and then it can be ignored.
    Since $v, u_1$ and $u_2$ form a triangle, we need to delete all the vertices in $N(\{v, u_1, u_2\})$ in this branch. Note that $\{v, u_1, v_2\}$ is a heavy triangle and we have that $|N(\{v, u_1, u_2\})|\geq 4$.
    We generate the following ${4\choose 2}$ + 1 branches
    \[
        \begin{split}
            & {\tt cPCP}(G \setminus \{v\},k - 1), \\
            & {\tt cPCP}(G \setminus (\{u_1, u_i\}),k - |\{u_1, u_i\}|) \mbox{~for each $i = 2, 3, 4$},\\
            & {\tt cPCP}(G \setminus (\{u_2, u_i\}),k - |\{u_2, u_i\}|) \mbox{~for each $i = 3, 4$},\\
            \mbox{and} \quad & {\tt cPCP}(G \setminus N(\{v, u_1, u_2\}),k - |N(\{v, u_1, u_2\})|).
        \end{split}
    \]

    \begin{figure}[!t]
        \centering
        \includegraphics[scale=0.4]{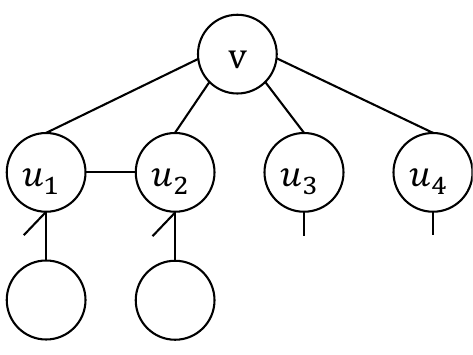}
        \caption{Degree-4 vertex $v$ is in a heavy triangle $\{v, u_1, u_2\}$.}
        \label{Step3}
    \end{figure}

The corresponding recurrence is
\[
    T (k) \leq T (k - 1) + ({4\choose 2} - 1) \times T (k - 2) + T(k - |N(\{v, u_1, u_2\})|) + 1,
\]
where $|N(\{v, u_1, u_2\})|\geq 4$.
For the worst case that $|N(\{v, u_1, u_2\})| = 4$, the branching factor of it is 2.8186.

\begin{lemma}\label{degree-1_lemma}
    If a vertex $v$ has two degree-1 neighbors $u_1$ and $u_2$, then there is a minimum cPCP-set containing none of $u_1$ and $u_2$.
    Furthermore, If $v$ is a vertex of degree at most 3, then there is a minimum cPCP-set containing none of $v$, $u_1$, and $u_2$.
\end{lemma}

\begin{proof}
    Let $S$ be a minimum cPCP-set. If $v\in S$, then $S\setminus \{u_1,u_2\}$ is still a minimum cPCP-set.
    Next, we assume that $v\not\in S$. If at least one of $u_1$ and $u_2$ is in $S$, then $S'=(S\setminus \{u_1,u_2\})\cup\{v\}$ is a cPCP-set with $|S'|\leq |S|$.
    Thus, $S'$ is a minimum cPCP-set not containing $u_1$ and $u_2$.

    If $v$ is a degree-2 vertex, then the component containing $v$ is a path of three vertices. None of the three vertices should be deleted.
    If $v$ is a degree-3 vertex, we let $u_3$ be the third neighbor of $v$. Note that at least one vertex in $\{v,u_1,u_2,u_3\}$ should be deleted and then any solution $S$ will contain at least one vertex in $\{v,u_1,u_2,u_3\}$. We can see that $S'=(S\setminus  \{v,u_1,u_2,u_3\})\cup \{u_3\}$ is still a cPCP-set
    with size $|S'|\leq |S|$. There is always a minimum cPCP-set containing none of $v$, $u_1$, and $u_2$.
\end{proof}

\vspace{2mm}
\noindent\textbf{Step 4} (Degree-4 vertices in a triangle).
    Assume that there is still a degree-4 vertex $v$ in a triangle $\{v, u_1,u_2\}$. We also let $u_1$, $u_2$, $u_3$ and $u_4$ be the four neighbors of $v$.
    Since triangle $\{v, u_1,u_2\}$ can not be a heavy triangle now, we know that $|N(\{v, u_1, u_2\})| \leq 3$.
    First, we show that it is impossible $|N(\{v, u_1, u_2\})| =2$. Assume to the contrary that $|N(\{v, u_1, u_2\})| =2$. Then $u_1$ and $u_2$  can only be adjacent to vertices in $N[v]$.
    If both of $u_1$ and $u_2$ are degree-2 vertices, then Reduction-Rule 1 should be applied. If one of $u_1$ and $u_2$ is a vertex of degree at least 3, then $v$ would dominate this vertex, and then the condition of Step 2 would hold. Any case is impossible.
    Next, we assume that $|N(\{v, u_1, u_2\})| = 3$ and let $u_5$ denote the third vertex in $N(\{v, u_1, u_2\})$. We further consider several different cases.

    Case 1: One of $u_1$ and $u_2$, say $u_1$ is a degree-4 vertex. For this case, vertex $u_1$ is adjacent to $u_5$, otherwise the degree-4 vertex $v$ would dominate the degree-4 vertex $u_1$.
     Since $u_1$ is of degree 4, we know that $u_1$ is also adjacent to one of $u_3$ and $u_4$, say $u_3$. Vertices $u_2$ and $u_3$ can only be adjacent to vertices in $N[v]\cup \{u_5\}$, otherwise $\{v,u_1,u_2\}$ or $\{v,u_1,u_3\}$ would form a heavy triangle and Step 3 should be applied. Thus, neither $u_2$ nor $u_3$ can be a degree-3 vertex, otherwise degree-4 vertex $u_1$  or $v$ dominates a degree-3 vertex $u_2$ or $u_3$ and then Step 2 should be applied.
        Thus, we have that either $d(u_2) = d(u_3) = 2$ or $d(u_2) = d(u_3) = 4$.
    We further consider the following three cases:

    Case 1.1: $d(u_2) = d(u_3) = 2$.
    For this case, we have that $v$ dominates $u_2$ and $u_3$.
    By Lemma \ref*{domination}, we have that there is a minimum cPCP-set either containing $v$ or containing none of $v$, $u_2$ and $u_3$.
    For the first branch, we delete $v$ from the graph and include it in the deletion set.
    For the second branch, we delete $u_1$ and $u_4$  from the graph and include it in the deletion set.
    The corresponding recurrence is
    \[
        T (k) \leq T (k - 1) + T (k - 2) + 1,
    \]
    the branching factor of which is 1.6181.

    Case 1.2: $d(u_2) = d(u_3) = 4$ and $u_2$ and $u_3$ are not adjacent.  For this case, both of $u_2$ and $u_3$ are adjacent to $u_4$ and $u_5$.
    Since $v$ does not dominate $u_4$, we know that $u_4$ is adjacent to at least one vertex out of $N[v]$.
    Since triangle $\{v, u_2,u_4\}$ is not a heavy triangle, we know that $u_4$ is not adjacent to any vertex other than $N[v]\cup \{u_5\}$.
    Thus, $u_4$ is also adjacent to $u_5$. Since the maximum degree of the graph is 4 now, we know that this component only contains six degree-4 vertices $N[v]\cup \{u_5\}$, which should be eliminated by Reduction-Rule 2.

    Case 1.3: $d(u_2) = d(u_3) = 4$ and $u_2$ and $u_3$ are adjacent. For this case, since $v$ does not dominate $u_3$, we know that $u_3$ is adjacent to at least one vertex out of $N[v]$. This vertex can only be $u_5$. 
    Thus, the four neighbors of $u_3$ are $v, u_1$, $u_2$, and $u_5$.
    Now $u_1$ is a degree-4 vertex dominating a degree-4 vertex $u_3$. Step 2 should be applied. Thus, this case is impossible.

    \begin{figure}[!t]
        \centering
        \includegraphics[scale=0.4]{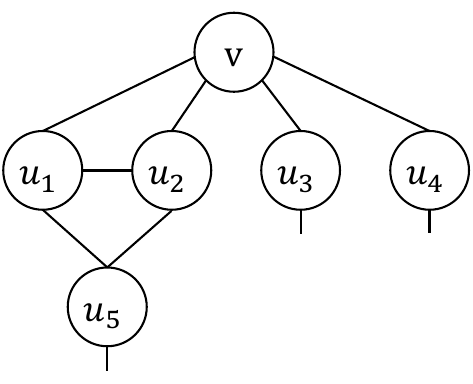}
        \caption{In the Case 2 of Step 4, degree-4 vertex $v$ is in a triangle $\{v, u_1, u_2\}$ and both of $u_1$ and $u_2$ are degree-3 vertices.}
        \label{Step4}
    \end{figure}
    Case 2:  Both of $u_1$ and $u_2$ are degree-3 vertices.
    It holds that $N(u_1) = \{v, u_2, u_5\}$ and $N(u_2) = \{v, u_1, u_5\}$, otherwise vertex $v$ would dominate $u_1$ or $u_2$.
    See Figure \ref*{Step4} for an illustration.
    For this case, we first branch on $v$ with Branching-Rule (B1) to generate ${d(v)\choose d(v) - 2} + 1$ branches.
    We can only get a recurrence relation
    \[
        T (k) \leq T (k - 1) + {4 \choose 2} T (k - 2) + 1.
    \]

    This recurrence is not good enough. Next, we look at the branch of deleting $v$ and try to get some improvements on this subbranch. After deleting $v$, vertices $u_1$ and $u_2$ become degree-2 vertices in a triangle.
    We first apply Reduction-Rule 1 to delete edge $u_1u_2$ between $u_1$ and $u_2$. Then, vertices $u_1$ and $u_2$ become degree-1 vertices adjacent to $u_5$.

    Case 2.1: $u_5$ is a degree-2 vertex. This case is impossible otherwise Reduction Rule \ref*{rrule-3} would be applied before this step.

    Case 2.2: $u_5$ is a degree-3 vertex. By Lemma \ref*{degree-1_lemma}, we can delete $N[u_5]$ directly and include the third neighbor of $u_5$ in the deletion set.
    We generate the following ${4\choose 2}$ + 1 branches
    \[
        \begin{split}
            & {\tt cPCP}(G \setminus (\{v\}\cup N[u_5]), k - 2) \\
            \mbox{and} \quad & {\tt cPCP}(G \setminus (N(v)\setminus \{u, w\}),k - |N(v)\setminus \{u, w\}|) \mbox{~for each pair $\{u, w\}\subseteq N(v)$}.
        \end{split}
    \]

    The corresponding recurrence is
    \[
        T (k) \leq T (k - 2) + {4 \choose 2} T (k - 2) + 1.
    \]
    which has a branching factor of 2.6458.

    Case 2.3: $u_5$ is a degree-4 vertex. By Lemma \ref*{degree-1_lemma}, we know that there is a minimum cPCP-set either containing $u_5$ or containing none of $u_5, u_1$, and $u_2$.
    For the first case, we delete $u_5$ and include it in the deletion set. For the second case, we delete $N[u_5]$ from the graph and include $N(u_5)\setminus \{u_5,u_1,u_2\}$ in the deletion set.
    Note that $|N(u_5)\setminus \{u_5,u_1,u_2\}|=2$ since $u_5$ is a degree-4 vertex.
    Combining with
    the previous branching on $v$, we get the following ${4\choose 2}$ + 2 branches
    \[
        \begin{split}
            & {\tt cPCP}(G \setminus \{v, u_5\}, k - 2), \\
            & {\tt cPCP}(G \setminus (\{v\}\cup N[u_5]), k - 3), \\
            \mbox{and} \quad & {\tt cPCP}(G \setminus (N(v)\setminus \{u, w\}),k - |N(v)\setminus \{u, w\}|) \mbox{~for each pair $\{u, w\}\subseteq N(v)$}.
        \end{split}
    \]

    The corresponding recurrence is
    \[
        T (k) \leq T (k - 1 - 1) +T(k - 1 - 2) + {4 \choose 2} T (k - 2) + 1,
    \]
  which has a branching factor of 2.7145.


After Step 4, no degree-4 vertex is in a triangle.

\vspace{2mm}
\noindent\textbf{Step 5} (Degree-4 vertices adjacent to some vertex of degree at least 3).
    Assume there is a degree-4 vertex $v$ adjacent to at least one vertex of degree at least 3.
    Let $u_1$, $u_2$, $u_3$ and $u_4$ be the four neighbors of $v$, where we assume without loss of generality that $d(u_1) \geq 3$.
    First, we branch on $v$ by either (b1) including it in the solution set or (b2) excluding it from the solution set. Next, we focus on the latter case (b2).
    In (b2), we further branch on $u_1$ by (b2.1) either including it in the solution set or (b2.2) excluding it from the solution set.
    For case (b2.1), there are at least $|N(v)|-2=2$ vertices in $N(v)$ that should be deleted by the same argument for Branching-Rule (B1). Thus, we can generate three subbranches by deleting $\{u_1, u_2\}$, $\{u_1, u_3\}$, and $\{u_1, u_3\}$, respectively.
    For case (b2.2),  there are still at least $|N(v)|-2=2$ vertices in $N(v)$ that should be deleted, which can be one of the following three sets $\{u_2, u_3\}$, $\{u_2, u_4\}$, and $\{u_3, u_4\}$.
    Furthermore, there are also at least $|N(u_1)| - 2\geq 1$ vertices in $N(u_1)$ that should be deleted. Note that $v\in N(u_1)$ is not allowed to be deleted now. Thus, at least $|N(u_1)\setminus\{v\}|-1$ vertices in $N(u_1)\setminus\{v\}$ should be deleted.
    We will generate ${{|N(u_1)\setminus\{v\}|}\choose {|N(u_1)\setminus\{v\}|-1}}=|N(u_1)\setminus\{v\}|=d(u_1)-1$ branches by decreasing $k$ by $|N(u_1)\setminus\{v\}|-1=d(u_1)-2$.
    Since after Step 3, the degree-4 vertex $v$ is not in any triangle, we know that $N(u_1)\setminus\{v\}$ is disjoint with $\{u_2,u_3,u_4\}$. For case (b2.2), we will generate $3\times (d(u_1)-1)$ subbranches by decreasing $k$ by at least $2+d(u_1)-2=d(u_1)$ in each, where $d(u_1)=3$ or 4.
    See Figure \ref*{Step5} for an illustration.
    In total, we will generate the following 1+3+$3\times (d(u_1)-1)$ branches
\[
    \begin{split}
        & {\tt cPCP}(G \setminus \{v\},k - 1), \\
        & {\tt cPCP}(G \setminus (\{u_1, u_i\}),k - 2) \mbox{~for each $i = 2, 3, 4$},\\
        & {\tt cPCP}(G \setminus (\{u_2, u_3\} \cup (N(u_1)\setminus \{v, w\})),k - d(u_1)),~\mbox{for each $w \in N(u_1)\setminus \{v\}$},\\
        & {\tt cPCP}(G \setminus (\{u_2, u_4\} \cup (N(u_1)\setminus \{v, w\})),k - d(u_1)),~\mbox{for each $w \in N(u_1)\setminus \{v\}$},\\
        \mbox{and}~ & {\tt cPCP}(G \setminus (\{u_3, u_4\} \cup (N(u_1)\setminus \{v, w\})),k - d(u_1)),~\mbox{for each $w \in N(u_1)\setminus \{v\}$}.
    \end{split}
\]

\begin{figure}[!t]
    \centering
    \includegraphics[scale=0.28]{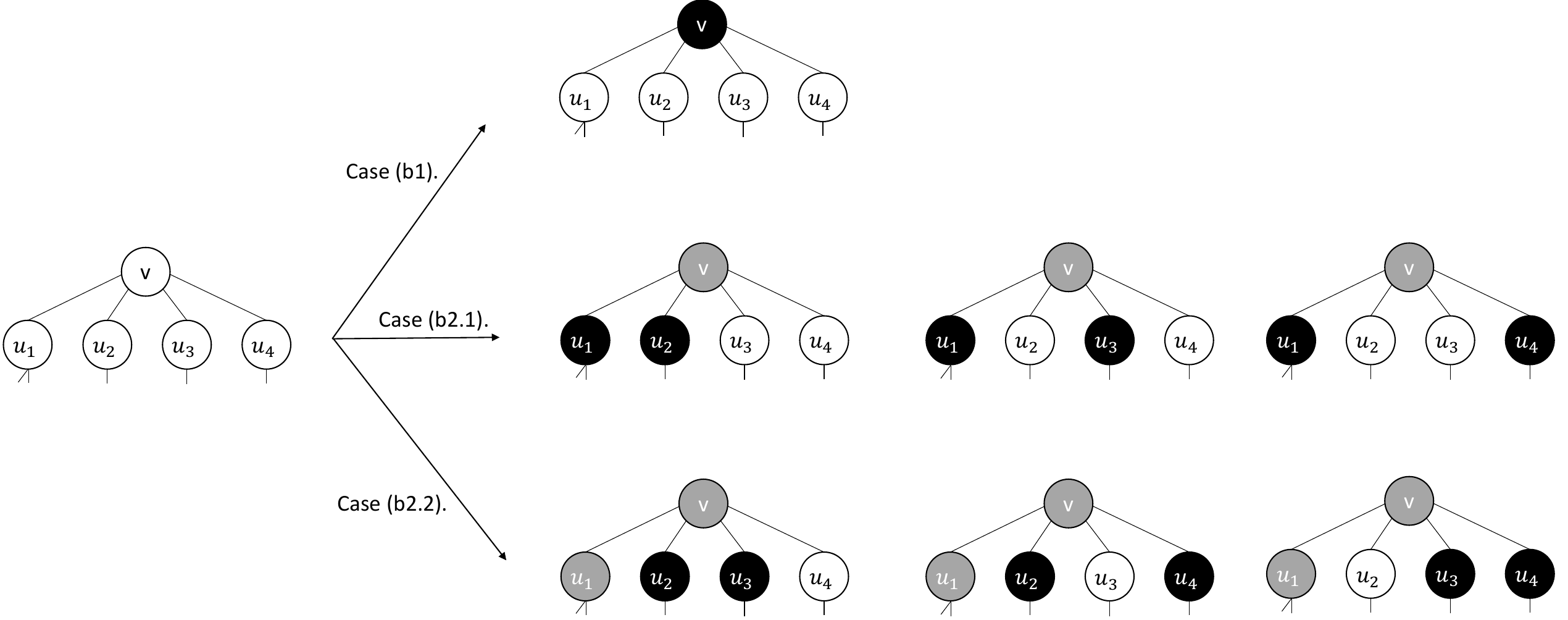}
    \caption{Vertices in the deletion set are denoted by black vertices, and vertices not allowed to be deleted are denoted by grey vertices.}
    \label{Step5}
\end{figure}

We get a recurrence
\[
    T (k) \leq T (k - 1) + 3 \times T (k - 2) + 3(d(u_1)-1) \times T(k - d(u_1)) + 1,
\]
where $d(u_1)= 3$ or 4.
For the case that $d(u_1) = 3$, the branching factor is 2.8192, and for the case that $d(u_1) = 4$, the branching factor is 2.6328.
\vspace{2mm}


The worst branching factors in the above five steps are listed in Table 2.
After Step 5, any degree-4 vertex can be only adjacent to vertices of degree at most 2. It is easy to see that the remaining graph after Step 5 is a proper graph.
We call the algorithm in Lemma \ref*{solve-proper} to solve the instance in $O^*(2.5199^{k})$ time.

\begin{table}[!t]\label{tb-result-L}
    \begin{center}
    \caption{The branching factors of each of the first five steps}
    \begin{tabular}{l|c|c|c|c|c}
       \hline
       Steps & Step 1 & Step 2 & Step 3 & Step 4 & Step 5 \\
       \hline
       Branching factors & 2.5445 & 2.3028 & 2.8186 & 2.7145 & \textbf{2.8192}\\
       \hline
    \end{tabular}
    \end{center}
    \vspace{-6mm}
\end{table}

\begin{theorem}\label{para-theorem}
    \textsc{Co-Path/Cycle Packing} can be solved in $O^*(2.8192^k)$ time.
\end{theorem}

Fomin et al.~\cite{fomin2019exact} introduced the monotone local search technique for deriving exact exponential-time algorithms from parameterized algorithms.
Under some conditions, a $c^k n^{O(1)}$-time algorithm implies an algorithm with running time $(2-1/c)^{n+o(n)}$.
By applying this technique on our $O^*(2.8192^k)$-time algorithm in Theorem \ref*{para-theorem}, we know that \textsc{Co-Path/Cycle Packing} can be solved in $(2-1/2.8192)^{n+o(n)} = O(1.6453^n)$ time.
\begin{corollary}
    \textsc{Co-Path/Cycle Packing} can be solved in $O(1.6453^n)$ time.
\end{corollary}

\section{A Parameterized Algorithm for \textsc{Co-Path Packing}}\label{cpp-section}

In this section, we show a randomized $O^*(2.9241^{k})$ time algorithm for \textsc{Co-Path Packing}. 
We also use the method for \textsc{Co-Path/Cycle Packing}, which consists of two phases. 
The first one is the dynamic programming phase, and the second one is the branch-and-search phase.

In the dynamic programming phase, it is more difficult to solve \textsc{Co-Path Packing} based on a path decomposition in comparison to \textsc{Co-Path/Cycle Packing}.
The reason is that \textsc{Co-Path/Cycle Packing} involves only local constraints, which means that the object's properties can be verified by checking each vertex's neighborhood independently.
For this problem, a typical dynamic programming approach can be used to design a $c^{pw(G)}|V(G)|^{O(1)}$ time algorithm straightforwardly.
In contrast, there is a global connectivity constraint for \textsc{Co-Path Packing}. 
The problem with global connectivity constraint is also called \emph{connectivity-type problem}.
For connectivity-type problems, the typical dynamic programming approach has to keep track of all possible ways the solution can traverse the corresponding separator of the path decomposition that is $\Omega(l^l)$, where $l$ is the size of the separator and hence the pathwidth~\cite{cygan2011solving}.

To obtain a single exponential algorithm parametrized by pathwidth for \textsc{Co-Path Packing}, we use the cut \& count framework.
Previously, cut \& count significantly improved the known bounds for various well-studied connectivity-type problems,
such as \textsc{Hamiltonian Path}, \textsc{Steiner Tree}, and \textsc{Feedback Vertex Set}~\cite{cygan2011solving}. 
Additionally, for \textsc{$k$-Co-Path Set}, cut \& count has been used to obtain a fast parameterized algorithm~\cite{sullivan2016fast}. 
In Section \ref*{cpp-dp-section}, we present a randomized fpt algorithm with complexity $O^*(5^{pw(G)})$ for \textsc{Co-Path Packing}.

In the branch-and-search phase, our algorithm for \textsc{Co-Path Packing} is similar to the branching algorithm for \textsc{Co-Path/Cycle Packing}.
Specifically, our branching algorithm for \textsc{Co-Path Packing} contains two reduction rules, two branching rules, and four steps, while the first reduction rule, both two branching rules, the first two steps, and the last step are the same as the branching algorithm for \textsc{Co-Path/Cycle Packing}.
The branch-and-search phase is shown in Section \ref*{cpp-algorithm-section}.

\subsection{A DP Algorithm via Cut \& Count for \textsc{Co-Path Packing}}\label{cpp-dp-section}

In this section, we use the cut \& count framework to design a $5^{pw}n^{O(1)}$ one-sided error Monte Carlo algorithm with a constant probability of a false negative for \textsc{Co-Path Packing}. %



\subsubsection{Cutting}
In this subsection, we first introduce the definitions of consistent cuts and marked consistent cuts.
Then we show the Isolation Lemma. 


\begin{definition}
    A partition $(V_1, V_2)$ of $V(G)$ is a \emph{consistent cut} of a graph $G$ if there is no edge $uv$ such that $u \in V_1$ and $v \in V_2$. Furthermore, all degree-0 vertices are contained in $V_1$.
\end{definition}

By the definition of consistent cut, each non-isolate connected component must be contained fully in $V_1$ or $V_2$.
So we have that a graph $G$ has exactly $2^{cc(G)-n_0(G)}$ consistent cuts, where $cc(G)$ is the number of connected components and $n_0(G)$ is the number of degree-0 vertices.
A \emph{marker set} is an edge set $M\subseteq E(G)$.

\begin{definition}
    A triple $(V_1, V_2, M)$ is a \emph{marked consistent cut} of a graph $G$ if $(V_1, V_2)$ is a consistent cut and $M \subseteq E(G[V_1])$.
    A marker set is \emph{proper} if it contains at least one edge in each non-isolate connected component of $G$. 
\end{definition}

Note that if a marked consistent cut contains a proper marker set, all vertices are contained in $V_1$.
The reason is that by the definition of marked consistent cut, any non-isolate connected component containing a marker must be contained in $V_1$.
Additionally, all degree-0 vertices are contained in $V_1$.
Therefore, If $M$ is a proper marker set, there exists exactly one corresponding consistent cut with $M$.
Conversely, if $M$ is not a proper marker set, there are an even number of consistent cuts with $M$, since there exists at least one unmarked non-isolate connected component that could be contained in either $V_1$ or $V_2$.

We call an induced subgraph $G'$ of $G$ a \emph{cc-solution} if $G'$ is a collection of disjoint paths.
\begin{definition}
    A pair $(G', M')$ is a \emph{marked-cc-solution} of a graph $G$ if $G'$ is a cc-solution and $M'$ is a proper marker set with size exactly equals to the number of non-isolate connected components of $G'$.
\end{definition}

If there exists a cc-solution with $|V(G)| - k$ vertices, then we can claim that there exists a cPP-set with size $k$.
Furthermore, if there exists a marked-cc-solution with $|V(G)| - k$ vertices, which means that there exists a cc-solution with $|V(G)| - k$ vertices, then we can also claim that there exists a cPP-set with size $k$.

A function $\omega : U \rightarrow \mathbb{Z}$ \emph{isolates} a set family $\mathcal{F} \subseteq 2^U$ if there is an unique $S' \in \mathcal{F}$ with $\omega(S') = \min_{S \in \mathcal{F}} \omega(S)$, where $\omega(X) = \sum_{x\in X}\omega(x)$.

\begin{lemma}[Isolation Lemma~\cite{mulmuley1987matching}]\label{isolation}
    Let $\mathcal{F} \subseteq 2^U$ be a set family over an universe $U$ with $|\mathcal{F}| > 0$. 
    For each $u \in U$, choose a weight $\omega(u) \in \{1, 2, \dots, N\}$ uniformly and independently at random. Then
    \[
        \Pr(\omega \text{ isolates } \mathcal{F}) \geq 1 - \frac{|U|}{N}.
    \]
    
\end{lemma}

\subsubsection{Counting}

\begin{definition}
    For some values $n \geq 0$, $e \geq 0$ and $m \geq 0$, A pair $(G', (V_1, V_2, M'))$ is a \emph{cc-candidate} of a graph $G$ if $G'$ is an induced subgraph of $G$ with maximum degree 2, exactly $n$ vertices and $e$ edges, and $(V_1, V_2, M')$ is a marked consistent cut of $G'$ with $m$ markers.
\end{definition}

Counting marked-cc-solutions remains difficult on a path decomposition since any feasible marked-cc-solution $(G', M')$ must include a cycle-free induced subgraph $G'$.
However, it is important to note that there is no global connectivity constraint for cc-candidates.
Therefore, counting cc-candidates is easier than counting marked-cc-solutions.
The following lemma shows that we can count cc-candidates instead of marked-cc-solutions in $\mathbb{Z}_2$.
For a graph $G$, assuming there is a weight function $\omega: V(G)\cup E(G) \rightarrow \mathbb{Z}$.
The weight of a marked-cc-solution $(G', M')$ is defined as $\sum_{v\in V(G')} \omega(v) + \sum_{e\in M'} \omega(e)$.
The weight of a cc-candidates $(G', (V_1, V_2, M'))$ is also defined as $\sum_{v\in V(G')} \omega(v) + \sum_{e\in M'} \omega(e)$. 


\begin{lemma}\label{parity}
    The parity of the number of marked-cc-solutions $(G', M')$ in $G$ with $n$ vertices, $e$ edges and weight $w$ is the same as the parity of the number of cc-candidates $(G', (V_1, V_2, M'))$ with $n$ vertices, $e$ edges, $n - e - n_0(G')$ markers, and weight $w$.
\end{lemma}

\begin{proof}

    Consider an induced subgraph $G'$ of $G$ and a marker set $M'$ of $G'$. Now we consider the following three cases.
    
    Case 1. $|V(G')|\neq n$ or $|E(G')|\neq e$ or $|M'|\neq n - e$ or $\sum_{v\in V(G')} \omega(v) + \sum_{m\in M'} \omega(m) \neq w$: 
    In this case, $(G', M')$ contributes 0 to both the number of marked-cc-solutions and the number of cc-candidates, respectively.

    Case 2. There is a cycle in $G'$:
    In this case, $(G', M')$ is not a feasible marked-cc-solution and contributes 0 to the number of marked-cc-solutions.
    Since there is a cycle in $G'$, we know that $cc(G') > n - e$. 
    The number of non-isolate connected components is greater than $n - e - n_0(G')$.
    So, there exists one non-isolate connected component containing no marker.
    Since this connected component could be contained in either $V_1$ or $V_2$, we have that there are an even number of feasible consistent cuts $(V_1, V_2)$ such that $(G', (V_1, V_2, M'))$ is a cc-candidate. 
    Thus, $(G', M')$ contributes an even number to the number of cc-candidates.

    Case 3. There is no cycle in $G'$:
    In this case, $G'$ is a feasible cc-solution.
    We further consider two cases.

    Case 3.1. $M'$ is not a proper marker set:
    In this case, $(G', M')$ is also not a feasible marked-cc-solution and contributes 0 to the number of marked-cc-solutions.
    By the definition of proper marker set, we have that there is one non-isolate connected component containing no marker.
    Following a similar argument in Case 2, we have that $(G', M')$ contributes an even number to the number of cc-candidates.

    Case 3.2. $M'$ is a proper marker set:
    In this case, $(G', M')$ is a feasible marked-cc-solution and contributes 1 to the number of marked-cc-solutions.
    By the definition of proper marker set, we have that every non-isolate connected component contains a marker.
    Therefore, each connected component can only be contained in $V_1$. 
    There is exactly one feasible consistent cut $(V_1, V_2)$ such that $(G', (V_1, V_2, M'))$ is a cc-candidate. 
    Thus, $(G', M')$ contributes 1 to the number of cc-candidates.

    Thus, For any subgraph $G'$ of $G$ and any marker set $M'$ of $G'$, $(G', M')$ either contributes 1 to both the number of marked-cc-solutions and cc-candidates, or contributes 0 to marked-cc-solutions and an even number to cc-candidates. This lemma holds.
\end{proof}

The following lemma shows the algorithm to count the number of cc-candidates in a given path decomposition.

\begin{lemma}\label{dp-cc-candidate}
    Given a graph $G$ and a path decomposition of $G$ with width $p$. 
    For any $a, n, e, w, m \geq 0$
    Determining the parity of the number of cc-candidates $(G', (V_1, V_2, M'))$ of $G$ with $a$ degree-0 vertices, $n$ vertices, $e$ edges, $m$ markers, and weight $w$ can be solved in $O^*(5^p)$ time and space.
\end{lemma}

\begin{proof}
Similar to the proof of Theorem \ref*{dp}. We simply assume that the path decomposition $P = (X_1, X_2,\cdots, X_r)$ is a nice path decomposition by Lemma \ref*{nice}.

Let $V_i = \bigcup_{j=1}^i X_i$ for each $i\in \{1,2,\cdots, r\}$. We have $V_r = V$.
For any $i\in \{1,2,\cdots, r\}$, let $\{D, R_0, R_1^1, R_1^2, R_2\}$ be an arbitrary partition of $X_i$.
We consider the following subproblem: to count the number of cc-candidates $(G', (V_1, V_2, M'))$ of the induced graph $G[V_i]$ such that 
$(G', (V_1, V_2, M'))$ is a feasible cc-candidate with $a$ degree-0 vertices, $n$ vertices, $e$ edges, $m$ markers, and weight $w$. 
And each vertex in $R_0$ is a degree-0 vertex in $V_1$; each vertex in $R_1^j$ is a degree-1 vertex in $V_j$ $(j = 1, 2)$; each vertex in $R_2$ is a degree-2 vertex in $V_1$ or $V_2$.
We also let $s(i, a, n, e, w, m, D, R_0, R_1^1, R_1^2, R_2)$ denote the corresponding number of the cc-candidates to this problem.
We only need to check the value $s(i = r, a, n, e, w, m, \emptyset, \emptyset, \emptyset, \emptyset, \emptyset)$ since $X_r = \emptyset$.
Next, we use a dynamic programming method to compute all $s(i, a, n, e, w, m, D, R_0, R_1^1, R_1^2, R_2)$.

For the case that $i = 1$, $X_1=\emptyset$ and it trivially holds that $s(1, 0, 0, 0, 0, 0, \emptyset, \emptyset, \emptyset, \emptyset, \emptyset) = 1$ and other input equals to 0. Since the path decomposition is nice, there are two cases for each $i\geq 2$.

    \vspace{1mm}
    \textbf{Case 1:} $X_i = X_{i-1} \cup v$ for some $v\notin X_{i-1}$.
    By the definition of path decomposition, we know that all neighbors of $v$ in the graph $G[V_i]$ must be in $X_{i-1}$.
    For every parameter combination $(a, n, e, w, m)$ and every partition $(D, R_0, R_1^1, R_1^2, R_2)$ of $X_i$, the following holds.
    \begin{enumerate}
        \item $ v\in D $. It holds $s(i,a,n,e,w,m,D,R_0,R_1^1,R_1^2,R_2) = s(i-1,a,n,e,w,m,D\setminus v,R_0,R_1^1,R_1^2,R_2)$.
        
        \item $ v\in R_1^j$ for some $j \in \{1, 2\}$. Let $R' = R_1^1\cup R_1^2\cup R_2$. We further consider two subcases.

        Case (a): $N_{G[V_i]}(v) \cap R_0 \neq \emptyset$ or $|N_{G[V_i]}(v) \cap R'| \neq 1$:

        For this case, the partition $(D, R_0, R_1^1, R_1^2, R_2)$ is infeasible and we simply have that 
        \[
            s(i,a,n,e,w,m,D,R_0,R_1^1,R_1^2,R_2)=0.
        \]

        Case (b): $N_{G[V_i]}(v) \cap R_0 = \emptyset$ and $|N_{G[V_i]}(v) \cap R'| = 1$.
        Let $u$ be the vertex contained in $N_{G[V_i]}(v) \cap R'$.
        We further consider two cases.
        
        Case (b.1): $v\in R_1^1$: 
        
        (i) If $u\in R_1^2$, the consistent cut corresponding to $(R_1^1, R_1^2)$ is infeasible and we simply have that 
        \[
            s(i,a,n,e,w,m,D,R_0,R_1^1,R_1^2,R_2)=0.
        \]
        
        (ii) If $u\in R_1^1$,
        we have that
        \[
        \begin{split}
            &s(i,a,n,e,w,m,D,R_0,R_1^1,R_1^2,R_2) = \\
            &s(i - 1,a + 1, n - 1,e - 1, w - w(v), m, D,R_0 \cup u,R_1^1\setminus \{u,v\},R_1^2,R_2) + \\
            &s(i - 1,a + 1, n - 1,e - 1, w - w(v) - w(uv), m - 1, D,R_0 \cup u,R_1^1\setminus \{u,v\},R_1^2,R_2).
        \end{split}
        \]
        (iii)  If $u\in R_2$,
        we have that
        \[
        \begin{split}
            &s(i,a,n,e,w,m,D,R_0,R_1^1,R_1^2,R_2) = \\
            &s(i - 1,a, n - 1,e - 1, w - w(v), m, D,R_0,R_1^1\setminus v \cup u,R_1^2,R_2\setminus u) + \\
            &s(i - 1,a, n - 1,e - 1, w - w(v) - w(uv), m - 1, D,R_0,R_1^1\setminus v \cup u,R_1^2,R_2\setminus u).
        \end{split}
        \]

        Case (b.2): $v\in R_1^2$: 
        
        (i) If $u\in R_1^1$, the consistent cut corresponding to $(R_1^1, R_1^2)$ is infeasible and we simply have that 
        \[
            s(i,a,n,e,w,m,D,R_0,R_1^1,R_1^2,R_2)=0.
        \]
        
        (ii) If $u\in R_1^2$,
        we have that
        \[
        \begin{split}
            &s(i,a,n,e,w,m,D,R_0,R_1^1,R_1^2,R_2) = \\
            &s(i - 1,a + 1, n - 1,e - 1, w - w(v), m, D,R_0 \cup u,R_1^1,R_1^2\setminus \{u,v\},R_2).
        \end{split}
        \]
        (iii)  If $u\in R_2$,
        we have that
        \[
        \begin{split}
            &s(i,a,n,e,w,m,D,R_0,R_1^1,R_1^2,R_2) = \\
            &s(i - 1,a, n - 1,e - 1, w - w(v), m, D,R_0,R_1^1,R_1^2\setminus v \cup u,R_2\setminus u).
        \end{split}
        \]

        \item $ v\in R_2$. Let $R' = R_1^1\cup R_1^2\cup R_2$. We further consider two subcases.
        
        Case (a): $N_{G[V_i]}(v) \cap R_0 \neq \emptyset$ or $|N_{G[V_i]}(v) \cap R'| \neq 2$:

        For this case, the partition $(D, R_0, R_1^1, R_1^2, R_2)$ is infeasible and we simply have that 
        \[
            s(i,a,n,e,w,m,D,R_0,R_1^1,R_1^2,R_2)=0.
        \]

        Case (b): $N_{G[V_i]}(v) \cap R_0 = \emptyset$ and $|N_{G[V_i]}(v) \cap R'| = 2$:
        Let $u_1, u_2$ be the two vertices contained in $N_{G[V_i]}(v) \cap R'$.
        If $u_1$ and $u_2$ are both contained in $V_1$, the edges $u_1v$ and $u_2v$ can be marked or not.
        Let $P = \{ (w - w(v), m),
        ( w - w(v) - w(u_1v), m - 1),
        ( w - w(v) - w(u_2w), m - 1),
        ( w - w(v) - w(u_1v) - w(u_2v), m - 2)\},$ 
        which is the possible parameter combination set for the case that $u_1$ and $u_2$ are contained in $V_1$. 
        Now we consider all possible cases for $\{u_1, u_2\}$.
        
        (i) If $u_1, u_2\in R_1^1$,
        we have that
        \[
        \begin{split}
            &s(i,a,n,e,w,m,D,R_0,R_1^1,R_1^2,R_2) = \\
            &\sum_{p\in P}s(i-1,a+2, n - 1, e - 2,(p), D,R_0\cup\{u_1,u_2\},R_1^1\setminus\{u_1,u_2\},R_1^2,R_2\setminus v).
        \end{split}
        \]
        (ii) If $u_1\in R_1^1$ and $u_2\in R_1^2$ or $u_1\in R_1^2$ and $u_2\in R_1^1$
        the consistent cut corresponding to $(R_1^1, R_1^2)$ is infeasible and we simply have that 
        \[
            s(i,a,n,e,w,m,D,R_0,R_1^1,R_1^2,R_2)=0.
        \]
        (iii) If $u_1, u_2\in R_1^2$ 
        we have that
        \[
        \begin{split}
            &s(i,a,n,e,w,m,D,R_0,R_1^1,R_1^2,R_2) = \\
            &s(i - 1,a + 2, n - 1, e - 2, w - w(v), m, D, R_0\cup\{u_1,u_2\}, R_1^1, R_1^2\setminus \{u_1,u_2\},R_2\setminus v).
        \end{split}
        \]
        (iv) If $u_1\in R_2$ and $u_2\in R_1^1$ or $u_2\in R_2$ and $u_1\in R_1^1$,
        we only consider the latter case without loss of generality.
        We have that 
        \[
        \begin{split}
            &s(i,a,n,e,w,m,D,R_0,R_1^1,R_1^2,R_2) = \\
            &\sum_{p\in P}s(i - 1,a + 1, n - 1, e - 2, (p), D, R_0 \cup u_1, R_1^1\setminus u_1 \cup u_2, R_1^2, R_2\setminus \{v, u_2\}).
        \end{split}
        \]
        (v) If $u_1\in R_2$ and $u_2\in R_1^2$ or $u_2\in R_2$ and $u_1\in R_1^2$,
        we only consider the latter case without loss of generality.
        We have that 
        \[
        \begin{split}
            &s(i,a,n,e,w,m,D,R_0,R_1^1,R_1^2,R_2) = \\
            &s(i - 1,a + 1, n - 1,e - 2, w - w(v), m, D, R_0\cup u_1, R_1^1, R_1^2\setminus u_1\cup u_2, R_2\setminus \{v, u_2\}).
        \end{split}
        \]
        (vi) If $u_1, u_2\in R_2$ 
        we have that 
        \[
            \begin{split}
                &s(i,a,n,e,w,m,D,R_0,R_1^1,R_1^2,R_2) = \\
                &s(i - 1,a, n - 1,e - 2, w - w(v), m, D,R_0,R_1^1,R_1^2\cup \{u_1,u_2\}, R_2\setminus \{u_1,u_2,v\})+\\
                &\sum_{p\in P}s(i - 1, a, n - 1, e - 2, (p), D, R_0, R_1^1\cup \{u_1,u_2\}, R_1^2, R_2\setminus \{u_1,u_2,v\})
            \end{split}
        \]

    \end{enumerate}


    \vspace{1mm}
    \textbf{Case 2:} $X_i = X_{i-1} \setminus v$ for some $v\in X_{i-1}$.
It is not hard to see that the following equation holds.
    \[
        \begin{split}
            s(i,a,n,e,w,m,D,R_0,R_1^1,R_1^2,R_2) =  &s(i-1,a,n,e,w,m,D\cup v, R_0, R_1^1, R_1^2, R_2)+\\
                                                    &s(i-1,a,n,e,w,m,D, R_0\cup v, R_1^1, R_1^2, R_2)+\\
                                                    &s(i-1,a,n,e,w,m,D, R_0, R_1^1\cup v, R_1^2, R_2)+\\
                                                    &s(i-1,a,n,e,w,m,D, R_0, R_1^1, R_1^2\cup v, R_2)+\\
                                                    &s(i-1,a,n,e,w,m,D, R_0, R_1^1, R_1^2, R_2\cup v).
        \end{split}
    \]

    For each bag $X_i$, there are at most $5^{|X_i|}$ different partitions, where $|X_i|\leq p + 1$. 
    The values $\{a,n,e,w,m\}$ are bounded by a polynomial of the graph size.
    For $i$ and each partition $\{D, R_0, R_1^1, R_1^2, R_2\}$ of $X_i$,
    it takes constant time to compute $s(i,a,n,e,w,m,D,R_0,R_1^1,R_1^2,R_2)$ by using the above recurrence relations.
    Therefore, our dynamic programming algorithm runs in $O(5^{p + 1}\cdot r\cdot \text{poly}(n) )$ time, where $r$ is bounded by a polynomial of the graph size.
    This lemma holds.
\end{proof}

Now we are ready to prove the main theorem in this section.

\begin{theorem}\label{dp-cPP}
    Given a path decomposition of $G$ with width $p$. \textsc{Co-Path Packing} can be solved in $O^*(5^p)$ time and space with failure probability $\leq 1/3$.
\end{theorem}

\begin{proof}
    Let $\mathcal{F}$ be the set of $(V', M')$ where $(G[V'], M')$ is a feasible marked-cc-solution.
    Let $U$ contain $V(G)$ and $E(G)$. 
    Then $2^U$ denotes all pairs of a vertex subset and an edge subset, which are potential marked-cc-solutions.
    Let $N = 3|U| = 3(|V(G)| + |E(G)|)$. 
    Each vertex and edge is assigned a weight in $[1, N]$ uniformly at random by $\omega$ and the probability of finding an isolating $\omega$ is 2/3 by Lemma \ref*{isolation}. 
    Based on such $\omega$, we have that there exists a weight $w$ such that there is exactly one marked-cc-solution with weight $w$ (if $\mathcal{F}\neq \emptyset$).
    It is easy to see that there exists a cPP-set with size at most $k$ if and only if there exists a marked-cc-solution with $n\geq |V(G)| - k$ and any possible $e$ and $w$.
    Thus, for any possible $e$ and $w$, by Lemma \ref*{parity}, we call the algorithm given in Lemma \ref*{dp-cc-candidate}, 
    then check the parity of the number of cc-candidates $(G', (V_1, V_2, M'))$ of $G$ with any possible $a \geq 0$ degree-0 vertices, $n\geq |V(G)| - k$ vertices, $e$ edges, $n - e - a$ markers and weight $w$ 
    to determine the existence of the number of marked-cc-solution $(G', M')$ with $n\geq |V(G)| - k$ and weight $w$.
    Thus, this lemma holds.
\end{proof}

\subsection{The Whole Algorithm}\label{cpp-algorithm-section}
In this section, we propose a parameterized algorithm for \textsc{Co-Path Packing}.
First, in Section 4.2.1, we show that \textsc{Co-Path Packing} on a proper graph class can be quickly solved by using the dynamic programming algorithm based on path decompositions in Theorem \ref*{dp-cPP}.
Similarly, we will use Lemma \ref*{get-pw} to conclude this result.
Second, in Section 4.2.2, we give a branch-and-search algorithm that will implement some branching steps on special local graph structures. 
Our branching algorithm for \textsc{Co-Path Packing} is similar to the branching algorithm for \textsc{Co-Path/Cycle Packing}.
Specifically, our branching algorithm for \textsc{Co-Path Packing} contains two reduction rules, two branching rules, and four steps, while the first reduction rule, both two branching rules, the first two steps and the last step are the same as the branching algorithm for \textsc{Co-Path/Cycle Packing}. 
When all the steps cannot be applied, we show that the graph must be a proper graph and then the algorithm in Section \ref*{cpp-dp-section} can be called directly to solve the problems.

\subsubsection{Proper Graphs with Small Pathwidth}
Recall that a graph is called proper if it satisfies the following conditions:
\begin{enumerate}
    \item The maximum degree of $G$ is at most 4.
    \item For any degree-4 vertex $v$, all neighbors are of degree at most 2.
    \item For any degree-2 vertex $v$, at least one vertex in $N(v)$ is of degree at least 3. 
    \item Each connected component contains at least 6 vertices.
\end{enumerate}

Combine Lemma \ref*{get-pw} for \textsc{Co-Path Packing} with Lemma \ref*{dp-cPP}, we have the following lemma.
\begin{lemma}\label{solve-proper-cPP}
    \textsc{Co-Path Packing} on proper graphs can be solved in $O^*(2.9241^{k})$ time with probability at least 2/3.
\end{lemma}

\begin{proof}
    We first call the algorithm in Lemma \ref*{get-pw}.
    If the algorithm decides that $G$ has no $cPP$-set of size at most $k$, we claim that $(G, k)$ is a no-instance.
    Otherwise, we can obtain a nice path decomposition of width at most $\frac{2k}{3} + \epsilon k$.
    Then, we call the algorithm in Theorem \ref*{dp-cPP}. This algorithm runs in $O^*(5^{2k/3+\epsilon k}) = O^*(2.9241^{k})$ with failure probability $\leq 1/3$, where we choose $\epsilon < 10^{-6}$.
    This lemma holds.
\end{proof}


\subsubsection{A Branch-and-Search Algorithm}
In this subsection, we provide a branch-and-search algorithm for \textsc{Co-Path Packing}, which is denoted by ${\tt cPP}(G, k)$.
Our algorithm contains several reduction and branching steps. 
After recursively executing these steps, we will get a proper graph and then call the dynamic programming algorithm in Lemma \ref*{solve-proper-cPP} to solve it.

\vspace{2mm}
\noindent\textbf{4.2.2.1. Reduction and Branching Rules.}

Firstly, we present two reduction rules. 
The first reduction rule is Reduction Rule \ref*{rrule-1} for \textsc{Co-Path/Cycle Packing} and Reduction Rule *2 is a new reduction rule for \textsc{Co-Path Packing}.

\vspace{2mm}
\noindent\textbf{Reduction-Rule 1.}
\textit{
    If there is a connected component $C$ of the graph such that $|V(C)|\leq 6$, then run a brute force algorithm to find a minimum cPP-set $S$ in $C$, delete $C$ and include $S$ in the deletion set.
}
\vspace{2mm}

A path $v_0v_1\dots v_{h-1}v_h$ is called a \emph{degree-two-path} if the two vertices $v_0$ and $v_h$ are of degree not 2 and the other vertices $v_1\dots v_{h-1}$ are of degree 2, where we allow $v_0 = v_h$.

\begin{lemma}\label{reduction-2-lemma-cPP}
    For a degree-two-path $P = v_0v_1\dots v_{h-1}v_h$ with $h \geq 4$ in $G = (V, E)$, let $G' = (V\setminus v_2, E\setminus\{v_1v_2, v_2v_3\} \cup \{v_1v_3\})$, 
    we have that $(G, k)$ is a yes-instance if and only if $(G', k)$ is a yes-instance.
\end{lemma}

\begin{proof}
    If one of $\{v_0, v_h\}$ is a degree-1 vertex, without loss of generality, we say $v_0$ is a degree-1 vertex.
    In this case, it is easy to see that vertices $\{v_0, v_1, \dots, v_{h-1}\}$ are not contained in any minimum solution, and this lemma holds.

    If $v_0$ is of degree at least 3, let $S$ be a minimum cPP-set for $G$. We consider the following cases.

    (i) $v_1\in S$: In this case, $v_2\notin S$ since $S$ is a minimum cPP-set. Vertex set $S$ is also a minimum cPP-set for $G'$.

    (ii) $v_1\notin S$: If $v_2\in S$, it is easy to check that $S' = S\setminus v_2 \cup v_1$ is a feasible cPP-set with size $|S|$.  Vertex set $S'$ is a minimum cPP-set for $G'$. If $v_2\notin S$, clearly $S$ is also a minimum cPP-set for $G'$.

    Let $S'$ be a minimum cPP-set for $G'$. 
    By the similar argument, there exists a minimum cPP-set $S''$ with size $|S'|$ for $G$.
\end{proof}

Based on this lemma, we have the following reduction rule for \textsc{Co-Path Packing}.

\vspace{2mm}
\noindent\textbf{Reduction-Rule *2.}
\textit{
    If there is a degree-two-path $P = v_0v_1\dots v_{h-1}v_h$ with $h \geq 4$, then
    return ${\tt cPP}(G' = (V\setminus v_2, E\setminus\{v_1v_2, v_2v_3\} \cup \{v_1v_3\}), k)$.
}
\vspace{2mm}

If Reduction Rules 1 and *2 cannot be applied, we have a property that for any degree-2 vertex $v$, at least one vertex in $N(v)$ is of degree at least 3.

After applying the three simple reduction rules, we will execute some branching steps.
For \textsc{Co-Path Packing}, we have two branching rules, which are the same as the two branching rules for \textsc{Co-Path/Cycle Packing}. 
The correctnesses of these two branching rules for \textsc{Co-Path Packing} are similar to the two branching rules for \textsc{Co-Path/Cycle Packing}.

\vspace{2mm}
\noindent\textbf{Branching-Rule (B1).}
\textit{
    For a vertex $v$ of degree at least 3, branch on it to generate ${|N(v)|\choose 2} + 1$ branches by
    either (i) deleting $v$ from the graph and including it in the deletion set,
    or (ii) for every pair of vertices $u$ and  $w$ in $N(v)$,
    deleting $N(v)\setminus \{u, w\}$ from the graph and including $N(v)\setminus \{u, w\}$ in the deletion set.
}
\vspace{2mm}

\vspace{2mm}
\noindent\textbf{Branching-Rule (B2).}
\textit{
    Assume that a vertex $v$ of degree at least 3 dominates a vertex u.
    Branch on $v$ to generate $1+ (|N(v)|-1)=|N(v)|$ branches by
    either (i) deleting $v$ from the graph and including it in the deletion set,
    or (ii) for each vertex $w\in N(v)\setminus \{u\}$,
    deleting $N(v)\setminus \{u, w\}$ from the graph and including $N(v)\setminus \{u, w\}$ in the deletion set.
}
\vspace{2mm}


\vspace{2mm}
\noindent\textbf{4.2.2.2. Steps.}

When we execute one step, we assume that all previous steps are not applicable in the current graph anymore.
In this subsection, we present four steps: Step 1, Step 2, Step *3 and Step *4 for \textsc{Co-Path Packing}.
Steps 1 and 2 are the Steps 1 and 2 for \textsc{Co-Path/Cycle Packing}.
Step *3 is a new step designed for \textsc{Co-Path Packing}.
Step *4 is the Step 5 for \textsc{Co-Path/Cycle Packing}.
 
\vspace{2mm}
\noindent\textbf{Step 1} (Vertices of degree at least 5).
    If there is a vertex $v$ of $d(v) \geq 5$, then branch on $v$ with Branching-Rule (B1) to generate ${d(v)\choose d(v) - 2} + 1$ branches.



\vspace{2mm}
\noindent\textbf{Step 2} (Degree-4 vertices dominating some vertex of degree at least 3).
    Assume that there is a degree-4 vertex $v$ that dominates a vertex $u_1$, where $d(u_1)\geq 3$. Without loss of generality, we assume that the other three neighbors of $v$ are $u_2, u_3$ and $u_4$, and $u_1$ is adjacent to $u_2$ and $u_3$ ($u_1$ is further adjacent to $u_4$ if $d(u_1)=4$). 
    We branch on $v$ with Branching-Rule (B2) and get $|N(v)|$ branches.



\vspace{2mm}
\noindent\textbf{Step *3} (Degree-4 vertices in a triangle).
    Assume that a degree-4 vertex $v$ is in a triangle.
    Let $u_1$, $u_2$, $u_3$ and $u_4$ be the four neighbors of $v$, where we assume without loss of generality that $\{v, u_1, u_2\}$ is a triangle. 
    We branch on $v$ with Branching-Rule (B1).
    In the branch of deleting $\{u_3,u_4\}$, we can simply assume that the three vertices $v, u_1$ and $u_2$ are not deleted, otherwise this branch can be covered by another branch and then it can be ignored.
    Since $v$, $u_1$ and $u_2$ form a triangle, we know this case is impossible, so we can ignore this subbranch.
    We generate the following ${4\choose 2}$ branches
    \[
        \begin{split}
            & {\tt cPP}(G \setminus \{v\},k - 1), \\
            & {\tt cPP}(G \setminus (\{u_1, u_i\}),k - |\{u_1, u_i\}|) \mbox{~for each $i = 2, 3, 4$},\\
            \mbox{and} \quad & {\tt cPP}(G \setminus (\{u_2, u_i\}),k - |\{u_2, u_i\}|) \mbox{~for each $i = 3, 4$}.
        \end{split}
    \]

The corresponding recurrence is
\[
    T (k) \leq T (k - 1) + ({4\choose 2} - 1) \times T (k - 2) + 1.
\]
The branching factor of it is 2.7913.
\vspace{2mm}

After Step *3, no degree-4 vertex is in a triangle.

\vspace{2mm}
\noindent\textbf{Step *4} (Degree-4 vertices adjacent to some vertex of degree at least 3).
    Assume there is a degree-4 vertex $v$ adjacent to at least one vertex of degree at least 3.
    Let $u_1$, $u_2$, $u_3$ and $u_4$ be the four neighbors of $v$, where we assume without loss of generality that $d(u_1) \geq 3$.
    First, we branch on $v$ by either (b1) including it in the solution set or (b2) excluding it from the solution set. Next, we focus on the latter case (b2).
    In (b2), we further branch on $u_1$ by (b2.1) either including it in the solution set or (b2.2) excluding it from the solution set.
    For case (b2.1), there are at least $|N(v)|-2=2$ vertices in $N(v)$ that should be deleted by the same argument for Branching-Rule (B1). Thus, we can generate three subbranches by deleting $\{u_1, u_2\}$, $\{u_1, u_3\}$, and $\{u_1, u_3\}$, respectively.
    For case (b2.2),  there are still at least $|N(v)|-2=2$ vertices in $N(v)$ that should be deleted, which can be one of the following three sets $\{u_2, u_3\}$, $\{u_2, u_4\}$, and $\{u_3, u_4\}$.
    Furthermore, there are also at least $|N(u_1)| - 2\geq 1$ vertices in $N(u_1)$ that should be deleted. Note that $v\in N(u_1)$ is not allowed to be deleted now. Thus, at least $|N(u_1)\setminus\{v\}|-1$ vertices in $N(u_1)\setminus\{v\}$ should be deleted.
    We will generate ${{|N(u_1)\setminus\{v\}|}\choose {|N(u_1)\setminus\{v\}|-1}}=|N(u_1)\setminus\{v\}|=d(u_1)-1$ branches by decreasing $k$ by $|N(u_1)\setminus\{v\}|-1=d(u_1)-2$.
    Since after Step *3, the degree-4 vertex $v$ is not in any triangle, we know that $N(u_1)\setminus\{v\}$ is disjoint with $\{u_2,u_3,u_4\}$. For case (b2.2), we will generate $3\times (d(u_1)-1)$ subbranches by decreasing $k$ by at least $2+d(u_1)-2=d(u_1)$ in each, where $d(u_1)=3$ or 4.


\vspace{2mm}


The worst branching factors in the above steps are listed in Table 3.
After Step *4, any degree-4 vertex can be only adjacent to vertices of degree at most 2. It is easy to see that the remaining graph after Step *4 is a proper graph.
We call the algorithm in Lemma \ref*{solve-proper-cPP} to solve the instance for \textsc{Co-Path Packing} in $O^*(2.9241^{k})$ time with probability at least 2/3.

\begin{table}[!t]\label{tb-result-L-cPP}
    \begin{center}
    \caption{The branching factors of each of the steps}
    \begin{tabular}{l|c|c|c|c|c|c}
       \hline
       Steps & Step 1 & Step 2 & Step *3  & Step *4 \\
       \hline
       Branching factors & 2.5445 & 2.3028 & 2.7913 & \textbf{2.8192}\\
       \hline
    \end{tabular}
    \end{center}
    \vspace{-6mm}
\end{table}

\begin{theorem}\label{para-theorem-cPP}
    \textsc{Co-Path Packing} can be solved in $O^*(2.9241^{k})$ time with probability at least 2/3.
\end{theorem}


\section{Conclusion}\label{conclusion-section}

In this paper, we show that 
given a path decomposition of width $p$, \textsc{Co-Path Packing} can be solved by a randomized fpt algorithm running in $O^*(5^p)$ time.
Additionally, by combining this algorithm with a branch-and-search algorithm, we show that \textsc{Co-Path/Cycle Packing} can be solved in $O^*(2.8192^k)$ time and \textsc{Co-Path Packing} can be solved in $O^*(2.9241^k)$ time with probability at least 2/3.
For \textsc{Co-Path/Cycle Packing}, the new bottleneck in our algorithm is Step 5, which is to deal with degree-4 vertices not in any triangle.
For \textsc{Co-Path Packing}, the new bottleneck in our algorithm is the dynamic programming phase.
The idea of using path/tree decomposition to avoid bottlenecks in branch-and-search algorithms may have the potential to be applied to more problems.
It would also be interesting to design a deterministic algorithm for \textsc{Co-Path Packing} faster than $O^*(3^k)$.


\bibliography{cPCP}

\newpage

\appendix

\centerline{\bf\large Appendix}
\section{Proofs} \label{PROOFS}

\noindent\textbf{Theorem \ref*{dp}}
\textit{
    Given a path decomposition of $G$ with width $p$. For any $d\geq 1$, \textsc{$d$-Bounded-Degree Vertex Deletion} can be solved in $O^*((d+2)^p)$ time and space.
}

\begin{proof}
We can simply assume that the path decomposition $P=(X_1, X_2,\cdots, X_r)$ is a nice path decomposition by Lemma \ref*{nice}.

Let $V_i = \bigcup_{j=1}^i X_i$ for each $i\in \{1,2,\cdots, r\}$. We have $V_r= V$.
For any $i\in \{1,2,\cdots, r\}$, let $\{D, R_0, R_1, \cdots, R_d\}$ be an arbitrary partition of $X_i$.
We consider the following subproblem: to find a minimum-size vertex set $S$ in the induced graph $G[V_i]$ such that $S\cap X_i=D$, the maximum degree of graph $G[V_i\setminus S]$ is at most $d$, and each vertex in $R_j$ $(j=0, i,\cdots,d)$ is a degree-$j$ vertex in $G[V_i\setminus S]$. We also let $s(i, D, R_0, R_1, \cdots, R_d)$ denote the corresponding size of the solution to this problem.
For some partitions $\{D, R_0, R_1, \cdots, R_d\}$ of $X_i$, there may not exist any feasible solution $S$, and we will let $s(i, D, R_0, R_1, \cdots, R_d)=\infty$ for this case.
To solve \textsc{$d$-Bounded-Degree Vertex Deletion}, we only need to check the minimum value among $s(i=r, D, R_0, R_1, \cdots, R_d)$ for all possible partitions $\{D, R_0, R_1, \cdots, R_d\}$ of $X_r$. 
Next, we use a dynamic programming method to compute all $s(i, D, R_0, R_1, \cdots, R_d)$.

For the case that $i = 1$, $X_1=\emptyset$ and it trivially holds that $s(1,\emptyset, \emptyset, \cdots, \emptyset) = 0$. Since the path decomposition is nice, there are two cases for each $i\geq 2$.

    \vspace{1mm}
    \textbf{Case 1:} $X_i = X_{i-1} \cup \{v\}$ for some $v\notin X_{i-1}$.
    By the definition of path decomposition, we know that all neighbors of $v$ in the graph $G[V_i]$ must be in $X_{i-1}$.
    For every partition $(D, R_0, R_1, \cdots, R_d)$ of $X_i$, the following holds.
    \begin{enumerate}
        \item $ v\in D $. It holds $s(i,D, R_0, R_1, \cdots, R_d) = 1 + s(i-1,D\setminus \{v\}, R_0, R_1, \cdots, R_d)$.
        \item $ v\in R_j$ for some $j \in \{0, 1, \cdots, d\}$. We further consider two subcases.

        Case (a): $N_{G[V_i]}(v) \cap R_0 \neq \emptyset$ or $|N_{G[V_i]}(v) \cap \bigcup_{l = 1}^{d} R_l| \neq j$.
        For this case, the partition $(D, R_0, R_1, \cdots, R_d)$ is infeasible and we can simply let $s(i,D, R_0, R_1$, $\cdots, R_d)=\infty$.

        Case (b): $N_{G[V_i]}(v) \cap R_0 = \emptyset$ and $|N_{G[V_i]}(v) \cap \bigcup_{l = 1}^{d} R_l| = j$.
        Let $ W = N_{G[V_i]}(v) \cap (\bigcup_{l = 1}^{d} R_l$), and $W_l = W\cap R_l$ for $l = 1, 2, \cdots, d$. Let $W_0 = W_{d + 1} = \emptyset$.
        We have that
        \[
            s(i,D, R_0, R_1, \cdots, R_d) = s(i-1,D, R_0', R_1', \cdots, R_j'\setminus \{v\}, \cdots, R_d'),
        \]
        where $R_l' = (R_l\setminus W_l) \cup W_{l + 1}$ for $l = 0, 1, \cdots, d$.

    \end{enumerate}


    \vspace{1mm}
    \textbf{Case 2:} $X_i = X_{i-1} \setminus \{v\}$ for some $v\in X_{i-1}$.
It is not hard to see that the following equation holds.
    \[
        \begin{split}
            s(i,D, R_0, R_1, \cdots, R_d) = \min\{&s(i-1,D\cup \{v\}, R_0, R_1, \cdots, R_d),\\
                                                  &s(i-1,D, R_0\cup \{v\}, R_1, \cdots, R_d),\\
                                                  &s(i-1,D, R_0, R_1\cup \{v\}, \cdots, R_d),\\
                                                  & \cdots\\
                                                  &s(i-1,D, R_0, R_1, \cdots, R_d\cup \{v\})\}.
        \end{split}
    \]

    For each bag $X_i$, there are at most $(d + 2)^{|X_i|}$ different partitions, where $|X_i|\leq p + 1$. For $i$ and each partition $\{D, R_0, R_1, \cdots, R_d\}$ of $X_i$,
    it takes at most $O(d)$ time to compute $s(i,D, R_0, R_1, \cdots, R_d)$ by using the above recurrence relations.
    Therefore, our dynamic programming algorithm runs in $O((d + 2)^{p + 1}\cdot r\cdot d )$ time, where $r$ is bounded by a polynomial of the graph size.
    This theorem holds.
\end{proof}

\end{document}